\documentclass[sigconf]{acmart}
\makeatletter                   %1
\def\mdseries@tt{m}             %1
\makeatother                    %1
\usepackage[plain]{fancyref}
\usepackage[draft=true]{minted} %2
\usepackage{color}
\usepackage{hyperref}           %5
\hypersetup{
    colorlinks=true,
    linkcolor=blue,
    filecolor=red,      
    urlcolor=magenta,
    breaklinks=true,            %3
}
\usepackage{breakurl}           %3
\copyrightyear{2019} 
\acmYear{2019} 
\setcopyright{rightsretained} 
\acmConference[KDD '19]{The 25th ACM SIGKDD Conference on Knowledge Discovery and Data Mining}{August 4--8, 2019}{Anchorage, AK, USA}
\acmDOI{10.1145/3292500.3330987}
\acmISBN{978-1-4503-6201-6/19/08}


\def\BibTeX{{\rm B\kern-.05em{\sc i\kern-.025em b}\kern-.08emT\kern-.1667em\lower.7ex\hbox{E}\kern-.125emX}}

\usepackage{booktabs} % For formal tables

\usepackage{amsmath}
\usepackage{algorithm}
\usepackage[noend]{algpseudocode}
\usepackage[utf8]{inputenc}
\usepackage{amsfonts}
\usepackage{amsthm}
\usepackage{xspace}
\usepackage{makecell} 
\usepackage{subfigure} 
\usepackage{balance} 
\usepackage{url} 

\usepackage{tikz}
\usetikzlibrary{shapes,snakes}
\usetikzlibrary{arrows,automata}

\newtheorem{theorem}{Theorem}[section]
\newtheorem{corollary}{Corollary}[theorem]
\newtheorem{lemma}[theorem]{Lemma}

\newcommand{\para}[1]{\smallskip \noindent {\bf #1.}}

\newcommand{\capped}{capped\xspace}
\newcommand{\caplet}{caplet\xspace}
\newcommand{\caplets}{caplets\xspace}
\newcommand{\Caplets}{Caplets\xspace}

\newcommand{\tsD}{{\textsf{t-Star-Decomposition}\xspace}}

\newcommand{\cC}{\ensuremath{\mathcal{C}}}

\newcommand{\cI}{\ensuremath{\mathcal I}}

\newcommand{\cP}{\ensuremath{\mathcal P}}

\newcommand{\bZ}{\ensuremath{\mathbb Z}}

\newcommand{\RRN}{\mathbb{R}_{\ge 0}}

\DeclareMathOperator*{\argmin}{arg\,min}

\newcommand{\dist}{\ensuremath{\mathit{dist}}\xspace}

\newcommand{\mymath}[1]{\newline\centerline{$\displaystyle{#1}$}}

\begin{document}

\title{Clustering without Over-Representation}

\author{Sara Ahmadian}
\affiliation{%
 \institution{Google Research}
 \city{New York}
 \state{NY}
 \country{US}}
\email{sahmadian@google.com}

\author{Alessandro Epasto}
\affiliation{%
 \institution{Google Research}
 \city{New York}
 \state{NY}
 \country{US}}
\email{aepasto@google.com}

\author{Ravi Kumar}
\affiliation{%
 \institution{Google Research}
 \city{Mountain View}
 \state{CA}
 \country{US}} 
\email{ravi.k53@gmail.com}

\author{Mohammad Mahdian}
\affiliation{%
 \institution{Google Research}
 \city{New York}
 \state{NY}
 \country{US}}
\email{mahdian@google.com}

\renewcommand{\shortauthors}{S. Ahmadian, A. Epasto, R. Kumar, M. Mahdian}

\begin{abstract}
In this paper we consider clustering problems in which each point is endowed with a color.  The goal is to cluster the points to minimize the classical clustering cost but with the additional constraint that no color is over-represented in any cluster.   This problem is motivated by practical clustering settings, e.g., in clustering news articles where the color of an article is its source, it is preferable that no single news source dominates any cluster.  

For the most general version of this problem, we obtain an algorithm that has provable guarantees of performance; our algorithm is based on finding a fractional solution using a linear program and rounding the solution subsequently.  For the special case of the problem where no color has an absolute majority in any cluster, we obtain a simpler combinatorial algorithm also with provable guarantees.  Experiments on real-world data shows that our algorithms are effective in finding good clustering without over-representation.
\end{abstract}

\maketitle

\section{Introduction}
Clustering is a fundamental problem in data mining and unsupervised machine learning. Many variants of this problem have been studied in the literature. In a number of applications, clustering needs to be performed in the presence of additional constraints, such as those associated with fairness or diversity. Chierichetti et al.~\cite{chierichetti2017fair} study one such clustering problem, where the constraint is that the distribution of a particular feature (say, gender) in each cluster is identical to that of the general population. This is a highly constraining requirement, particularly in cases where the protected feature can take many values, and in many cases such a clustering does not exist. Furthermore, in many applications, such as the ones explained below, proportional representation is not really required: a clustering that ensures  no particular feature value is highly over-represented in any cluster suffices.

A motivating application for our work is the following: every day, online advertising systems sell billions of advertising opportunities, specified by keywords the advertisers provide, through auctions. This is a highly heterogeneous set of auctions, and to optimize any of the auction parameters, one needs to cluster this set into smaller, more homogeneous, clusters. However, to ensure that no advertiser can manipulate this process, it is crucial that no advertiser has a large market share in any cluster (see~\cite{WWW18} for a theoretical justification of this statement). Hence, keywords must be clustered such that no advertiser is over-represented in any cluster. 

In addition to the above, there are other settings where an upper bound on the representation of each group in each cluster can capture real-world requirements. For example, in clustering news articles, requiring that no cluster is dominated by a certain view point or a certain news source is a good way to ensure balance and diversity in each cluster. Another example is clustering a number of agents into committees, where it is desirable that no committee is dominated by agents of a certain background. See~\citet{celis2018multiwinner} for an example where a similar constraint is applied to the problem of selecting a single committee maximizing a certain scoring function.

\para{Our contributions}
In this paper we formulate the problem of clustering without over-representation and study its algorithmic properties. For the clustering part, we focus on the $k$-center formulation. While there are many different well-studied models for clustering (such as $k$-median, $k$-means, $k$-center, or correlation clustering), we have picked the $k$-center model because of its theoretical simplicity (which allows us to prove good theoretical bounds) as well as the strong guarantees that are useful in many applications (that every point in a cluster is close to the center of that cluster).

Our formulation of the problem is in terms of a parameter $\alpha$ that specifies the maximum fraction of nodes in a cluster that have a specific value for the protected feature. Our main results are the following. First, for the case of $\alpha = 1/2$, we obtain a combinatorial approximation algorithm. Note that $\alpha = 1/2$ is a canonical case as it corresponds to ensuring that no cluster is dominated by a group with an absolute majority. Second, for the case of general $\alpha$, we give an approximation algorithm based on linear programming (LP) that achieves a bicriteria approximation guarantee. We also prove that the problem is NP-hard to approximate. Finally, we evaluate our LP-based algorithm on a number of real data sets, showing that its performance is even better than the theoretical guarantees.

\para{Related work}
Clustering is a classical problem in unsupervised learning and finds application in a variety of settings (see, e.g.,~\citet{jain2010data}); examples include information retrieval, image segmentation, and targeted marketing. The most popular clustering formulation studies the problem under an optimization objective that minimizes the $\ell_p$ norm for $p\in\{1,2,\infty\}$ corresponding to $k$-median, $k$-means, and $k$-center, respectively. In this work, our focus is on the $k$-center case, which admits a $2$-approximation~\cite{gonzalez1985clustering, hochbaum1985best} and is NP-hard to approximate within a factor better than $2$~\cite{hsu1979easy}. 

Fairness in machine learning is relatively new but has received a significant amount of attention. This includes research on defining notions of fairness~\cite{calders2010three, dwork2012fairness, feldman2015certifying, kamishima2012fairness} and on designing algorithms that respect fairness~\cite{celis2018multiwinner, celis2017ranking,chierichetti2017fair, joseph2016fairness,kamishima2012fairness, yang2017measuring,backursfair}.  A recent line of work considers batch classification algorithms that achieve group fairness or equality of outcomes and avoid disparate impact~\cite{calders2010three, feldman2015certifying, kamishima2011fairness, fish2016confidence}.

\citet{chierichetti2017fair} extended the notion of disparate impact to clustering problems and studied the fair $k$-center problem in the case there are only two groups (also called colors). This was later generalized by ~\citet{rosner2018privacy} to multiple groups, achieving a $14$-approximation algorithm in the general case. Even with two colors, the problem is challenging, and the optimum solution can violate common conventions, e.g., a point may not necessarily be assigned to the closest open center. 
The main difference between our work and that of~\cite{chierichetti2017fair,rosner2018privacy} is that the latter focuses on the problem of finding a clustering where the distribution of colors is in each cluster is exactly the same as the distribution of colors over all given data points, whereas we only require that in each cluster, the fraction of nodes of each color is at most a given threshold. Note that requiring exact proportional representation in each cluster is often prohibitively restrictive. For example, if the number of times different colors appear in the graph are relatively prime, there is no non-trivial feasible clustering in the setting of \citet{chierichetti2017fair}, whereas our formulation often admits non-trivial solutions.  

Concurrently and independently, Bera et al.~\cite{bera2019fair} and Bercea et al.~\cite{bercea2018fair} obtained algorithms to convert an arbitrary clustering solution to a fair one, sacrificing both approximation and fairness.  They provide bicriteria approximations for a more general problem (with upper and lower bound on the representation of a color). 
Our algorithm, however, is simpler and we prove (at most) an additive 2 violation for the fairness constraint (improved to 1 for a special case) in contrast to Bera et al.~\cite{bera2019fair} who prove an additive 4 violation and Bercea et al.~\cite{bercea2018fair} who do not bound the additive violation. 

There has been some work on clustering with diversity~\cite{liyizhang}, where the objective is to ensure each cluster has at least a certain number of colors; our objective is clearly different from this.  The large body of work on clustering with constraints~\cite{constrainedclustering}, to the best of our knowledge, does not address the over-representation constraint.

\para{Outline of the paper}
In Section~\ref{sec:prelim} we formalize the problem of finding an $\alpha$-capped $k$-center clustering.  In Section~\ref{sec:general} we present our main theoretical result, an LP-based algorithm for the general $\alpha$ case. Later, in Section~\ref{sec:half} we provide a purely combinatorial algorithm for the $\alpha=\frac{1}{2}$ case.  Then in Section~\ref{sec:experiments} we report the results of our empirical studies.  In Section~\ref{sec:hardness} we show that obtaining a decomposition in $\alpha$-capped clusters of minimum cost is hard for $\alpha \le 1/2$ irrespective of the constraint on the number of clusters.  Finally, in Section~\ref{sec:conclusion} we discuss future avenues of research.
\section{Model and Preliminaries}
\label{sec:prelim}

In the \emph{$k$-clustering} problem, we are given a set $D$ of points in a metric space with the distance function $d(\cdot, \cdot)$ and an integer bound $k$, and the goal is to cluster the points into at most $k$ clusters $\cC_1, \cC_2, \ldots, \cC_{k}$. Various clustering problems have been studied and in this paper, we focus on $k$-center clustering. We define the problem in terms of {\em facility location} terminology where points are referred to as {\em clients} and clusters are defined by the assignment of clients to {\em centers} (also called {\em facilities}). An instance $\cI=(D, F, d, k)$ of \emph{$k$-center} consists of a client set $D$, a facility set $F = D$, a metric space $d: D\times D\rightarrow \RRN $, and a positive integer bound $k$. A feasible $k$-center solution is a pair $(F', \sigma)$, where $F'\subseteq F$ is a set of at most $k$
facilities and $\sigma: D\rightarrow F'$ is a mapping that assigns each client $j$ to a facility $\sigma(j)\in F'$. The goal is to find a feasible solution that minimizes the {\em maximum radius} or {\em clustering cost} defined as $\lambda(F', \sigma) = \max_{j\in D} d(j, \sigma(j))$. Of course, in the classic $k$-center problem, once the set $F'$ is determined, assigning each client to the closest facility in $F'$ yields the assignment with minimum objective.  With additional constraints, however, the closest assignment might be infeasible.

Even though the standard $k$-center problem is computationally hard, it admits an elegant 2-approximation algorithm~\cite{hochbaum1985best}: first select an arbitrary point as center, then, iteratively pick the next center to be the point that is farthest from all currently chosen centers, until  $k$ centers are chosen.  For completeness, we present it below (Algorithm~\ref{alg:greedy}).  
\begin{algorithm}
\caption{Greedy-$k$-center($\cI = (D, D\subseteq F, d, k\ge 1)$).}\label{alg:greedy}
\begin{algorithmic}[1]
\State $i_0 \gets$ an arbitrary client in $D$.
\State $F' = \{i_0\}$
\For {$l \in \{1, 2, \ldots, k-1\}$}
\State $i_l \gets \arg\max_{j\in D} \min_{i \in F'}d(j, i)$, the furthest client from $F'$
\State $F' \gets F' \cup \{i_l\}$
\EndFor
\State $\forall j\in D: \sigma(j) \gets i = \argmin_{i\in F'} d(i,j)$
\State $\lambda \gets \lambda(F', \sigma)$
\State \Return $((F', \sigma), \lambda)$
\end{algorithmic}
\end{algorithm}
In this paper, we consider the {\em $\alpha$-capped $k$-center} problem where points have colors and we have a constraint on the representation of each color in each cluster. More precisely, in an $\alpha$-capped $k$-center instance $\cI = (D, F, d, k, \alpha, c)$, in addition to the input of classical $k$-center, we are given a fractional bound $\alpha \in (0, 1]$ and a color $c(j)$  for each point $j \in D$. We use $D_c$ to denote the set of clients of color $c$.  A feasible solution $(F', \sigma: D\rightarrow F' )$ is a feasible $k$-center solution that satisfies the {\em representation constraint}, which states that for each color $c$ and each facility $i$, the total number of clients of color $c$ assigned to $i$ should be no more than $\alpha$ fraction of all clients assigned to $i$. This constraint can be written as  
\mymath{
\forall i\in F', c:~~|\sigma^{-1}(i) \cap D_c| \leq \alpha |\sigma^{-1}(i)|.
}
The goal in $\alpha$-capped $k$-center problem is to find a feasible solution $(F', \sigma)$ that minimizes 
\mymath{
\lambda(F', \sigma) = \max_{j \in D} d(j, \sigma(j)).
}
Let $(F^*, \sigma^*)$ be the optimal clustering, and let $\lambda^* = \lambda(F^*, \sigma^*)$ be the optimal clustering cost.  A {\em $\rho$-approximation algorithm}, for $\rho \geq 1$, outputs a clustering $(F', \sigma)$ such that $\lambda(F', \sigma) \leq \rho \cdot \lambda(\sigma^*,C^*)$.  
\section{A general algorithm}
\label{sec:general}

We present a general algorithm to solve the $\alpha$-\capped $k$-center clustering problem.  The main idea is to first solve a linear program (LP) relaxation of the problem to obtain a fractional solution and then modify the fractional solution---sacrificing a little both in the approximation factor and in the representation constraint---to get an integral solution. In the course of doing this, we will get what is called a \emph{bicriteria} algorithm, i.e., while we get a constant-factor approximation to $\alpha$-\capped $k$-center, our solution will violate the $\alpha$ upper bound mildly. In fact, we can show that for each color and each facility, there are at most two extra clients in addition to the allowed number of clients, so the cap is violated additively by at most two additional nodes---a negligible quantity for a large cluster. 
%We present the details (Algorithm~\ref{alg:fair}) at the end of this section after we explain all the steps.

\subsection{An LP formulation}
For a given distance $\lambda \in \RRN$, consider the problem of finding a feasible assignment of clients to facilities in such a way that the clustering cost of the solution is at most $\lambda$. This problem can be formulated using the following integer program (IP).
{\small
\begin{eqnarray}
&\sum_{i\in F} x_{ij} \geq 1 &\forall j \in D, \label{ip_client_fully_assigned}\\
& x_{ij} \leq y_i &\forall i\in F, j\in D, \label{ip_facility_open_if_client_assgined}\\
&\sum_{j\in D_c} x_{ij} \leq \alpha \cdot \sum_{j\in D} x_{ij} &\forall c\in [t] , i\in F,\label{ip_facility_color_threshold}\\
&\sum_{i\in F} y_i \leq k, &\label{ip_total_open}\\
& x_{ij}, y_i\in \{0,1\} &\forall i \in F, j \in D, \label{ip_integrality}\\
&  x_{ij} = 0  &\forall i\in F, j \in D, d(i,j) > \lambda. \label{ip_x_0}
\end{eqnarray}
}
Here, the indicator variable $y_i$ denotes if facility $i \in F$ is open or not and the indicator variable $x_{ij}$ denotes if client $j$ is assigned to facility $i$.  Note that by constraint~\eqref{ip_x_0}, $x_{ij}$ can take non-zero value only if facility $i\in F$ is at distance at most $\lambda$ from client $j\in D$.  Constraint~\eqref{ip_facility_open_if_client_assgined} captures that a facility must be open if it has a client assigned to it,~\eqref{ip_facility_color_threshold} captures the representation constraint, and~\eqref{ip_total_open} captures that the total number of open facilities is at most $k$.

Before relaxing the integrality constraint of the above IP, we strengthen it by adding the following constraint: if a facility $i$ is open, it has to serve at least $\lceil\frac{1}{\alpha}\rceil$ clients to satisfy the representation constraint. Therefore, every integral solution of the above program must satisfy the inequality $\sum_{j\in D} x_{ij} \geq \lceil\frac{1}{\alpha}\rceil \cdot y_i$.

We consider the following LP obtained by adding this constraint and relaxing the integrality constraint~\eqref{ip_integrality}. We use $\cP(\lambda, \alpha)$ to denote the polytope defined by this LP.
{\small
\begin{eqnarray}
&\sum_{i\in F} x_{ij} \geq 1 &\forall j \in D, \nonumber \\
& x_{ij} \leq y_i &\forall i\in F, j\in D, \nonumber \\
&\sum_{j\in D_c} x_{ij} \leq \alpha \cdot \sum_{j\in D} x_{ij} &\forall c\in [t] , i\in F,\label{lp_facility_color_threshold}\\
& \sum_{j\in D} x_{ij} \geq \lceil\frac{1}{\alpha}\rceil \cdot y_i &\forall i\in F, \nonumber \\
&\sum_{i\in F} y_i \leq k, \nonumber \\
& 0 \leq y_i \leq 1 &\forall i \in F, \nonumber \\
& 0 \leq x_{ij} \leq 1 &\forall i \in F, j \in D, \nonumber \\
&  x_{ij} = 0  &\forall i\in F, j \in D, d(i,j) > \lambda. \nonumber
\end{eqnarray}
}
As mentioned above, we present a bicriteria algorithm that finds a solution that might violate the representation constraint, i.e., constraint~(\ref{lp_facility_color_threshold}). We use the notation $\cP (\lambda, \alpha, \Delta)$, for $\Delta \in \RRN$, to denote the set of points that satisfy all the constraint for $\cP(\lambda, \alpha)$ except constraint~(\ref{lp_facility_color_threshold}) and only violate that constraint with an additive error of $\Delta$, i.e., $\sum_{j\in D_c} x_{ij} \leq \alpha \cdot \sum_{j\in D} x_{ij} + \Delta$. Note that $\cP(\lambda, \alpha) = \cP(\lambda, \alpha, 0)$. 

\subsection{Outline}

Recall that $\lambda^*$ is the value of the optimal solution to the problem.  The main idea in our algorithm is that, since the polytope $\cP(\lambda^*, \alpha)$ is non-empty, by binary search, we can first find the smallest value $\lambda'$ such that $\cP(\lambda', \alpha)$ is non-empty (since the set of distances between pairs of points is finite). Note that the non-emptiness check via solving the LP also yields a point $(x, y) \in \cP(\lambda', \alpha)$, which is a fractional solution to the LP.  The plan then is to use $(x, y)$ to construct a feasible integral solution in a slightly larger polytope, namely, $(x'', y'') \in \cP(3 \lambda',\alpha, 2)$, where $x'', y''$ are integral and hence will correspond to a valid solution to the $k$-center problem.
\begin{theorem}
\label{thm:general}
Given an instance $\cI$ of $\alpha$-\capped $k$-center clustering, there is a polynomial time algorithm that finds a solution $(F', \sigma)$ of  cost at most $3\lambda^*$ such that  
$$
|\sigma^{-1}(i) \cap D_c| \leq \alpha \cdot |\sigma^{-1}(i)| + 2.
$$
In the case of $1/\alpha \in \bZ$, we can actually improve the additive term to $1$ and in term of multiplicative bound we get $|\sigma^{-1}(i) \cap D_c| \leq  2\alpha |\sigma^{-1}(i)|$.
\end{theorem}
To prove Theorem~\ref{thm:general}, the integral solution $(x'', y'')$ is constructed from $(x, y)$ in two steps.  In the first step, we construct a solution $(x', y') \in \cP(3\lambda', \alpha)$ using $(x, y)$, where $y'$ is integral.  This step can be thought of as determining which facilities to open based on the fractional solution.  In the second step, we construct an integral solution $(x'', y'')\in \cP(3\lambda', \alpha, 2)$.  This step uses the open facilities to define a suitable maximum flow problem to obtain an assignment of clients to  facilities.  We describe these two steps.

\subsection{Finding facilities to open}\label{subset:findfac}

The goal in this step is to find $(x', y') \in \cP(3\lambda', \alpha)$ where $y'$ is integral.  Let $F'\subseteq F$ be a maximal subset of facilities such that any two facilities $i,i'\in F'$ are at least distance $2\lambda'$ from each other, i.e., $d(i,i') > 2\lambda'$.
%(Note that as stated in Section~\ref{sec:prelim}, the facility set and client set are the same, i.e., there is a client at each facility point). 
We open all facilities in $F'$, i.e., set  $y'_i = 1$ for $i \in F'$ and $y'_i = 0$ for $i\notin F'$. Note that if $\lambda'$ is a correct guess of the optimum, none pair of clients at locations in $F'$ can be served by the same center and so the size of $F'$ is smaller than or equal to $k$. Next, we show how to define $x'$. We essentially transfer the fractional assignment of clients from $F$ to $F'$.  First we define a mapping $\theta:\{i\in F ~\mid~ y'_i>0\} \rightarrow F'$ as
\begin{itemize}
        \item If $i\in F'$, then $\theta(i) =i$.
        \item If $i\notin F'$, then $\theta(i) = i'$ where $i' \in F'$ with $d(i,i') < 2\lambda'$.  (Such an $i'$ exists by the maximality of $F'$.)
\end{itemize} 
Now for each client $j \in D$, we can define 
\begin{equation*}
x'_{ij} = \left\{
\begin{array}{cc}
\sum_{i'\in \theta^{-1}(i)} x_{i'j} & i \in F' \\
0 & \mbox { otherwise. } 
\end{array}
\right.
\end{equation*}
We now show that $(x', y')$ has the desired properties.
\begin{lemma}
    $(x',y') \in \cP(3\lambda', \alpha)$ and $y'$ is integral.
\end{lemma}
\begin{proof}
Let us first show that $x'_{ij}$ can only take non-zero value if facility $i$ is at distance $3\lambda'$ from it. If $x'_{ij}$ is non-zero, then there exists a facility $i'$ where $\theta(i')=i$ and $x_{ij} > 0$. Since $x_{ij} > 0$, we get that $d(i',j) < \lambda'$ and since $\theta(i') = i$, $d(i,i') < 2\lambda'$, so by the triangle inequality, we have $d(j,i) \leq 3\lambda'$. Since $x'$ is just rerouting the assignment of clients from facilities in $F$ to $F'$, $y_i = 1$ for all facilities in $F'$, and $F'$ has at most $k$ facilities, $(x',y')$  satisfy Constraints ~\eqref{ip_client_fully_assigned}, \eqref{ip_facility_open_if_client_assgined}, and \eqref{ip_total_open}. Constraint \eqref{ip_facility_color_threshold} is satisfied since for each $i\in F', c\in [t]$,
$$
\sum_{j\in D_c} x'_{ij} = \sum_{i'\in \theta^{-1}(i)}\sum_{j\in D_c} x_{i'j} \leq \alpha \cdot\sum_{i'\in \theta^{-1}(i)}\sum_{j\in D} x_{ij} = \alpha \cdot \sum_{j\in D} x'_{ij},
$$
where the inequality follows from the definition of $\theta$.
\end{proof}
    
\subsection{Assigning clients to facilities}\label{subsec:clientassgn}

The goal in this step is to construct a solution $(x'', y'') \in \cP(3\lambda', \alpha, 2)$ such that $x''$, $y''$ are integral.  In fact, $x''_{ij} > 0$ only if $x'_{ij} > 0$.
We let $(x'', y'')$ be the solution to the following maximum flow problem and use the fact that a network with integral bound on edges and integral demands, if feasible, always has an integral solution. 

Construct a flow network $(V, A)$ as follows:
\begin{itemize}
        \item $V = \{s,t\} \cup D \cup \{(i,c) ~\mid~ i\in F', c\in [t]\}$.
        \item $A = A_1 \cup A_2 \cup A_3 \cup A_4$ where $A_1=\{(s,j) ~\mid~ j\in D\}$ with capacity $1$, $A_2=\{(j,(i,c)) ~\mid~ j\in D_c, x'_{ij}>0\}$ with capacity $1$, $A_3=\{((i,c),i)\}$ with lower bound $\lfloor\sum_{j\in D_c}x'_{ij}\rfloor$ and capacity $\lceil\sum_{j\in D_c}x'_{ij}\rceil$, and $A_4=\{(i,t)\}$ with lower bound $\lfloor \sum_{j\in D} x'_{ij}\rfloor$ and capacity $\lceil \sum_{j\in D} x'_{ij}\rceil$.
\end{itemize}
Note that $(x',y')$ is a feasible flow of value $|D|$, so there is an integral flow of value $|D|$ such that a client $j$ sends a flow to a facility $i$ if $x'_{ij} > 0$. Thus $x''_{ij} > 0$ only if client $j$ is at distance $3\lambda'$ from facility $i$. This concludes the steps of our algorithm (Algorithm~\ref{alg:fair}).
\begin{algorithm}
\caption{Fair-$k$-center($\cI =(D, F, d, k), \alpha, \lambda$).}\label{alg:fair}
\begin{algorithmic}[1]
\State $(x, y) \gets \text{a feasible solution of } \cP(\lambda, \alpha)$
\If {$\cP(\lambda, \alpha)$ is empty}
\State\Return $(\emptyset, \emptyset)$
\EndIf
\State $F'\gets $ a maximal subset of $F$ where $\forall i\neq i'\in F': d(i,i') > 2\lambda$
\State $(x',y') \gets$ client reassignment based on $F'$ (Section~\ref{subset:findfac})
\State $(x'', y'') \gets $ client assignment based on max flow in network $(V,A)$ (Section~\ref{subsec:clientassgn})
\State $F^s \gets \{i ~\mid~ y''_i > 0\}$
\State $\forall j\in D$: $\sigma^s(j) \gets i$ where $x''_{ij} > 0$
\State\Return $(F^s, \sigma^s)$
\end{algorithmic}
\end{algorithm}
It remains to bound the violation of the representation constraint.
\begin{lemma}\label{lem:alpha_violation}For any color $c$ and any facility $i$, $\sum_{j\in D_c} x''_{ij} \leq \alpha \cdot \sum_{j\in D} x''_{ij} + 2$ where the additive term can be improved to $+1$ for $1/\alpha \in \bZ^+$.
\end{lemma}
\begin{proof}
Let $T' = \sum_{j\in D_c} x'_{ij}$, $B' = \sum_{j\in D} x'_{ij}$,  $T'' = \sum_{j\in D_c} x''_{ij}$, and $B'' = \sum_{j\in D} x''_{ij}$.  Since $(x',y')$ is a feasible solution of $\cP(\lambda', \alpha)$, we have $T'\leq \alpha\cdot B'$. Using the lower bounds and upper bounds on the edge $((i,c), i) \in A_3$, we know that $\lfloor T'\rfloor \leq T'' \leq \lceil T' \rceil$ and $\lfloor B'\rfloor \leq B'' \leq \lceil B' \rceil$. Since $\lceil T' \rceil < T' + 1$, we can bound $T''$ in terms of $B''$ as follows: 
$$
T'' < T' + 1 \leq \alpha B' + 1 \leq \alpha B'' + \alpha + 1 \leq \alpha B'' + 2.
$$
Now suppose $\alpha = 1/m$ for some $m\in \bZ^+$ and suppose $B'' = p\cdot m + r$ for $r < m$.  Then, $\alpha B'' + \alpha = p + \frac{r+1}{m}$. If $r < m - 1$, then the largest integer smaller than $\alpha B'' + \alpha + 1$ is $p + 1 \leq \alpha B'' + 1$. If $r = m - 1$, then $\alpha B'' + \alpha + 1 = p + 2$, now since $T'' < p + 2$, it follows that $T'' \leq p + 1 \leq \alpha B'' + 1$.
\end{proof}
We can bound the cost of the solution, in terms of violating the representation constraint multiplicatively as follows.
\begin{corollary}
For any color $c$ and facility $i$,
$\frac{
\sum_{j\in D_c} x''_{ij}}{ \sum_{j\in D} x''_{ij}} \leq  2\alpha
$ for $1/\alpha\in \bZ^+$.
\end{corollary}
\begin{proof}
Since $B'' \geq\lfloor B' \rfloor \geq\lfloor \frac{1}{1/m} \rfloor = m$, the $+1$ term in the last line of the proof of Lemma \ref{lem:alpha_violation}, can be bounded by $\alpha B''$.
\end{proof}

\renewcommand{\dist}{\mathrm{dist}}
\newcommand{\diam}{\mathrm{diam}}
\newcommand{\CK}{\mathcal{K}}
\newcommand{\tCK}{\tilde{\mathcal{K}}}

\section{An Algorithm for $\alpha = 1/2$}
\label{sec:half}

In this section, we present a simple, combinatorial approximation algorithm for the important special case of $\alpha = 1/2$.  This case corresponds to finding a clustering of the points such that no color is the absolute majority in any cluster, i.e., every color in a cluster occurs at most half of the times as the cluster size.  To proceed, we need two notions, namely, \caplets and  threshold graphs.  

\para{\Caplets}
Let $G$ be any graph whose set of nodes is $D$.  A \emph{caplet} in $G$ is a subset $K \subseteq D, 2 \le |K| \leq 3$ with distinct colors, i.e., $c(j) \neq c(j')$ for $j \neq j' \in K$.  Since \caplets can have either size two or three, we call the former case an \emph{edge \caplet} and the latter a \emph{triangle \caplet}.  For two \caplets $K_1$ and $K_2$, let $\dist(K_1, K_2)$ be defined as the minimum distance between pair of points of the two \caplets, i.e., $\dist(K_1, K_2) = \min_{j_1\in K_1, j_2 \in K_2} d(j_1, j_2)$.   Note that the distance function defined on \caplets is not necessarily a metric but will be useful to bound the distance between points belonging to different \caplets.  The \emph{diameter} of a \caplet $K$ is $\diam(K) = \max_{j, j' \in K} d(j, j')$.  The diameter of a set $\CK$ of \caplets is $\diam(\CK) = \max_{K \in \CK} \diam(K)$. 

A \emph{\caplet decomposition} $\kappa(G)$ of a connected graph $G$, if it exists, is a set of edge \caplets and at most one triangle \caplet such that each node in $G$ is present in exactly one \caplet.  Note that the only time when a \caplet decomposition uses a triangle \caplet is when the number of nodes in $G$ is odd. The caplet decomposition can be found in polynomial time by guessing the triangle if the size of graph is odd, and then finding the perfect matching on the remaining vertices. 

\para{Threshold graph}
Given $D$, a threshold $\tau > 0$, we define a \emph{threshold graph} $G(\tau) = (D,E)$ to be an undirected graph on the points in $D$, where $(j, j') \in E$ iff they have different colors and they are at distance at most $\tau$ from each other, i.e., $c(j) \neq c(j')$ and $d(j,j') \le \tau$. 
\subsection{Algorithm}

First of all, we assume that we know the optimal value $\lambda^* = \lambda(\sigma^*)$.  This is without loss of generality since by definition of $k$-center, $\lambda^* \in \{  d(i,j) ~\mid~ i \in F, j \in D \}$.  Hence an algorithm can enumerate over the set of all possible values for $\lambda^*$; at worst, this enumeration only costs an additional factor $|F|\cdot|D|$ in the running time.\footnote{One can also get an $1+\epsilon$ approximation of the optimum $\lambda^*$ in logarithmic many tries with standard techniques.}  Assuming we know $\lambda^*$, the idea is to create the threshold graph with $2\lambda^*$ as the threshold, and then to decompose it into \caplets.  Finally, the \caplets can be clustered using the greedy algorithm for $k$-center.  The steps are presented in Algorithm~\ref{alg:alg_0.5}.  
\begin{algorithm}
\caption{Non-dominant-$k$-center($\cI = (D, F, d, k), \alpha = 1/2$).}\label{alg:alg_0.5}
\begin{algorithmic}[1]
\For {$\lambda \in \{d(i,j) ~\mid~ i\in F, j\in D\}$ in non-decreasing order}
\State $D'\gets \emptyset$
\For {Connected component $C$ of $G(2\lambda)$}
\State $G_C \gets (C, E')$ where $E' = \{ (j,j')\;|\; c(j)\neq c(j'),\; d(j,j') \leq 10\lambda \}$
\If {no \caplet decomp. for $G_C$}
\State reject $\lambda$ and continue to next $\lambda$
\EndIf
\State $D' \gets D' \cup \left \{j_K \mid \text{arbitrary client} \; j_K \in K \in  \kappa(G_C) \right \}$
\EndFor
\State $((F^g, \sigma^g), \lambda^g) \gets \text{Greedy-$k$-center}(\cI'= (D', F, d, k))$
\If {$\lambda^g > 2\lambda$}
\State reject $\lambda$ and continue to next $\lambda$
\EndIf
\State $\forall j: \sigma^s(j) \gets \sigma^g(j_K)$ where $j, j_K \in K$.
\State \Return $(F^g, \sigma^s)$
\EndFor
\end{algorithmic}
\end{algorithm}

Note that our approach is similar in spirit to the fairlet decomposition approach proposed in~\cite{chierichetti2017fair}.  However, since our representation constraint is less stringent than the fair clustering constraint, as we will see, the reasoning becomes more delicate and involved.   

To show that Algorithm~\ref{alg:alg_0.5} obtains a provably good approximation, we show a key characterization: there is a \caplet decomposition of each connected component of $G(2\lambda^*)$ with small diameter.
\begin{lemma}\label{claim_fairlet_decom}
For each connected component $C$ of $G(2\lambda^*)$, there is a \caplet decomposition $\kappa(C)$ such that $\diam(\kappa(C)) \leq 10\lambda^*$.
\end{lemma}
Before proving the lemma, we use it to show that Algorithm~\ref{alg:alg_0.5} gives a good approximation.
\begin{theorem}
Algorithm \ref{alg:alg_0.5} finds a $(1/2)$-capped $k$-clustering solution of cost at most $12\lambda^*$.
\end{theorem}
\begin{proof}
Using Lemma~\ref{claim_fairlet_decom}, we know that the {\bf if} statement (line 3 in Algorithm~\ref{alg:alg_0.5}) fails for $\lambda^*$.  Furthermore, since the optimal \capped clustering yields a feasible solution for the $k$-center instance $I$ and there is a 2-approximation algorithm for $k$-center, a feasible solution can be found for $\lambda = \lambda^*$ (line 7).  Therefore the loop terminates successfully (line 8) for some $\lambda \leq \lambda^*$. 

We next show we get a valid $(1/2)$-capped clustering.  For each color $c$, note that the number of points of color $c$ assigned to facility $i \in F$ is at most the number of \caplets assigned to $i$.  However, by definition, each \caplet is of size at least two and has distinct colors.  Therefore, no color can be the absolute majority for each $i \in F$; this proves the $(1/2)$-capped property.  The cost of clustering is a $12$-approximation since each point $j$ in a \caplet $K$ assigned to a facility $i$ is at most at distance $2\lambda$ from $j_K$ and $d(j_K, j) \leq 10\lambda$ since $\diam(K) \leq 10 \lambda$ by Lemma~\ref{claim_fairlet_decom}.  The proof is complete as $\lambda \leq \lambda^*$.  
\end{proof}

\subsection{Analysis}

We now prove Lemma~\ref{claim_fairlet_decom}. Let $C$ be a connected component of $G(2\lambda^*)$.  There are two steps in the proof.  In the first step, we find a set $\CK_i$ of \caplets with respect to each facility $i$ such that $\diam(\kappa(\CK_i)) \leq 2\lambda^*$.  In the second step, we collect the \caplets $\kappa(\CK_i)$ for each $i \in F$ from the first step and appropriately modify them to obtain a \caplet decomposition $\kappa(\CK)$ of $C$ such that $\diam(\kappa(\CK)) \leq 10\lambda^*$.  (If we naively take the union of the \caplets for $i \in F$ we may not get a valid \caplet decomposition of $C$ since we might have more than one triangle \caplet, violating the definition.)

The first step is relatively straightforward.   Indeed, consider the optimal solution with open facilities $F^*$ and an assignment $\sigma: D\rightarrow F^*$.  Since for each open facility $i\in F^*$, the number of points with the same color is less than half of the points assigned to $i$, if $|\sigma^{-1}(i)|$ has even size, we can define a matching between points of different colors in $\sigma^{-1}(i)$.  If $|\sigma^{-1}(i)|$ is odd, then there are at least three colors present in $\sigma^{-1}(i)$. Define the triangle to include three points of  different colors and the rest of points in $\sigma^{-1}(i)$ can be matched to points of different colors.  This yields $\CK_i$ with the property that it has at most one triangle \caplet.  Furthermore that since all the points in $\sigma^{-1}(i)$ are at distance at most $\lambda^*$ from $i$, by the triangle inequality, any two points in $\sigma^{-1}(i)$ are at distance at most $2\lambda^*$ from each other.  Therefore, these points will belong to the same connected component of $G(2\lambda^*)$.  Let $\tCK_C = \cup_{i \in F^*} \CK_i  \cap C$.  

Next, we consider the second step.  For this, it is helpful to work with the graph $G' = (\tCK_C, E)$ such that for $K, K' \in \tCK_C$, we have $(K, K') \in E$ if $\dist(K, K') \leq 2\lambda^*$. Notice that $G'$ is connected since it is constructed from $C$. 

The goal is to transform the \caplets obtained in the first step into a valid \caplet decomposition of $C$.  This is done by finding a path between two triangle \caplets and ``shifting'' points to get a new set of edge \caplets, sacrificing some in the distance between \caplets.  Fix $C$ henceforth.

From $C$, we construct a set $P$ of disjoint paths with the following properties: each path in $P$ is of the form $K_0, \ldots, K_\ell$ where
(i) $K_0$ and $K_\ell$ are triangle \caplets and $K_i, 1 \leq i < \ell$ are edge \caplets, (ii) $\dist(K_i, K_{i+1}) \leq 6 \lambda^*$, and (iii) $\diam(K_i) \leq 2\lambda^*$.  Let $T$ be a minimal rooted tree spanning the nodes corresponding to triangle \caplets in $C$.   Note that all the leaves in $T$ correspond to triangle \caplets and the internal nodes in $T$ may be edge or triangle \caplets.  We perform a bottom-up procedure on $T$, removing paths from $T$ and adding them $P$ in an iterative manner; the procedure ends when $T$ has at most one triangle \caplet.  Let $T_f$ denote the rooted subtree of $T$ rooted at a node $f$. In the bottom-up procedure, we maintain the property that for each scanned node $f$ there is at most one triangle \caplet in $T_f$.  Note this property is already satisfied at the leaves. Let $K$ be the deepest node in the current tree that does not satisfy this property. If $K$ has more than one child, let $p_1 = (K, K_1, \ldots, K_r)$ and $p_2 = (K, K'_1, \ldots, K'_s)$ be two paths starting at $K$ and ending at triangle \caplets $K_r$ and $K_s$. Note that the degree of internal nodes on $p_1$ and $p_2$ is exactly two by the choice of $K$.  We add the path $p = (K_r, K_{r-1}, \ldots, K_1, K'_1, K'_2, \ldots, K'_s)$ to $P$ and remove the edges of $p_1 \cup p_2$ from $T$.  Since $K_1$ and $K'_1$ are at distance $2\lambda^*$ from $K$ and points inside $K$ are at distance at most $2\lambda^*$ from each other, $K_1$ and $K'_1$ are at distance at most $6\lambda^*$ from each other.  We continue this procedure until $K$ has at most one child.  If $K$ is a leaf, then we remove if it is an edge \caplet and leave it in $T$ if otherwise.  Else, let $K'$ be the sole child of $K$.  If $K'$ is an edge \caplet, we remove $K'$ from $T$.  If both $K$ and $K'$ are triangle \caplets, we add them to $P$ and remove both of them from $T$.  We continue the procedure until we reach the root and at the end of this, there exists at most one triangle \caplet that is not covered by a path in $P$.  It is also easy to see that each path in $P$ satisfies the desired properties.  (See Figure~\ref{fig:path}.)
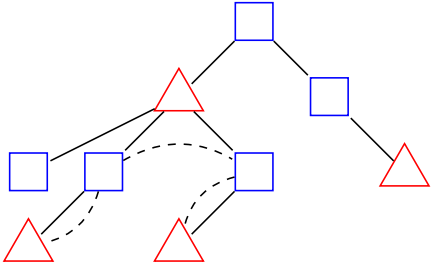
\begin{figure}
\centering
\begin{tikzpicture}[-,>=stealth',shorten >=1pt,auto,node distance=1.0cm, scale=0.20, semithick]
\node[draw, rectangle, minimum size=0.5cm] (a) [color=blue] {};
\node[draw, rectangle, minimum size=0.5cm] (b) [color=blue, below of = a, right of = a] {};
\node[draw, regular polygon,regular polygon sides=3, minimum size=0.75cm] (c) [color=red, below of = a, left of = a] {};
\node[draw, rectangle, minimum size=0.5cm] (d) [color=blue, below of = c, right of = c] {};
\node[draw, rectangle, minimum size=0.5cm] (e) [color=blue, below of = c, left of = c] {};
\node[draw, rectangle, minimum size=0.5cm] (e0) [color=blue, left of = e] {};
\node[draw, regular polygon,regular polygon sides=3, minimum size=0.75cm] (f) [color=red, below of = c, right of = e] {};
\node[draw, regular polygon,regular polygon sides=3, minimum size=0.75cm] (g) [color=red, below of = e, left of = e] {};
\node[draw, regular polygon,regular polygon sides=3, minimum size=0.75cm] (h) [color=red, below of = b, right of = b] {};

\path[]
  (a) edge node {} (b)
  (a) edge node {} (c)
  (c) edge node {} (d)
  (c) edge node {} (e)
  (c) edge node {} (e0)
  (d) edge node {} (f)
  (e) edge node {} (g)
  (h) edge node {} (b)
  (e) edge [dashed, bend left, above] node {} (g)
  (e) edge [dashed, bend left, below] node {} (d)
  (d) edge [dashed, bend right, below] node {} (f)
  ;
\end{tikzpicture}
\caption{Construction of a path (dashed line) in $P$.  The triangle nodes are triangle \caplets and the square nodes are edge \caplets.}
\label{fig:path}
\end{figure}

Now consider each $p = (K_0, K_1, \ldots, K_\ell) \in P$.  Recall from property (i) above that $K_0$ and $K_\ell$ are triangle \caplets and the rest are edge \caplets.  We define a new set of edge \caplets $K'_0, \ldots, K'_{\ell+1}$ as follows.  We pick an arbitrary point $i_0$ from $K_0$ and shift it to the next \caplet $K_1$ and then shift some point from $K_1$ to the next \caplet, and so on. More precisely, let $i_0$ be an arbitrary point in $K_0$, define $K'_0 = K_0 \setminus \{ i_0 \}$ and let $K'_1 = \{ i_0, i'_1\}$ where $i'_1$ is point in $K_1$ with different color than $i_0$.  We continue the process iteratively, where at each step $r$, we define the edge \caplet $K '_{r+1}$ to contain point $i_r$ in $K_r$ not covered by $K'_0, K'_1, \ldots, K'_{r}$ for ($r < \ell$), and point $i'_{r+1}$ in $K_{r+1}$ with different color than $i_r$. In the last step, a point $i'_{\ell-1} \in K_{\ell-1}$ is shifted and matched to a point $i_\ell \in K_\ell$ and we define $K'_{\ell+1} = K_\ell \setminus \{ i_\ell \}$.  Note this process is possible since each \caplet $K_r$ contains at least two points of different colors, there always exists a point that has a different color than the shifted point.  (See Figure~\ref{fig:shift}.)  By properties (ii) and (iii) above, the diameter of each \caplet is at most $2\lambda^*$ and two consecutive \caplets are at distance at most $6\lambda^*$ from each other.  Applying the triangle inequality, we get that the diameter of the \caplets in $K'_0, \ldots, K'_{\ell+1}$ is at most $10 \lambda^*$.  

\begin{figure}[htb]
\centering
\begin{tikzpicture}[-,>=stealth',shorten >=1pt,auto,node distance=1.3cm, , scale=0.20, semithick]
\node[state,color=red] (a) {};
\node[state,color=blue] (b) [right of = a] {};
\node[state,color=green] (c) [above of = b] {};
\node[state,color=blue] (d) [right of = b] {};
\node[state,color=green] (e) [above of = d] {};
\node[state,color=green] (f) [right of = d] {};
\node[state,color=red] (g) [above of = f] {};
\node[state,color=red] (h) [right of = f] {};
\node[state,color=blue] (i) [above of = h] {};
\node[state,color=green] (j) [right of = h] {};

\path[]
  (a) edge node {} (b)
  (b) edge node {} (c)
  (a) edge node {} (c)
  (d) edge node {} (e)
  (f) edge node {} (g)
  (h) edge node {} (i)
  (h) edge node {} (j)
  (i) edge node {} (j)

  (a) edge [bend left, above, dashed] node {} (c)
  (b) edge [dashed] node {} (e)
  (d) edge [dashed] node {} (f)
  (g) edge [dashed] node {} (i)
  (h) edge [bend left, below, dashed] node {} (j)
  ;
\end{tikzpicture}
\caption{The shifting operation in action on a path of four \caplets, beginning and ending with a triangle \caplet.  The solid lines denote the original \caplets and the dotted lines denote the new \caplets after the shifting operation.
\label{fig:shift}
}
\end{figure}
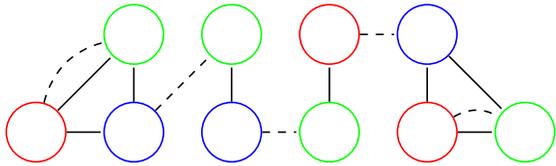

\newcommand{\reuters}{\texttt{reuters}\xspace}
\newcommand{\victorian}{\texttt{victorian}\xspace}
\newcommand{\fourarea}{\texttt{4area}\xspace}
\newcommand{\query}{\texttt{query}\xspace}

\section{Empirical evaluation}
\label{sec:experiments}

In this section we empirically evaluate our algorithms on several publicly-available datasets from the UCI Repository%
\footnote{\url{http://archive.ics.uci.edu/ml}}%
and DBLP%
\footnote{\url{http://dblp.uni-trier.de/xml/}}%
, as well as on a proprietary dataset related to online auctions.  In our empirical analysis we focus on the LP-based algorithm (Section~\ref{sec:general}).  We describe the datasets used, the baselines we consider, the quality measures we compute, and finally the results. 

\subsection{Datasets}
The datasets reported in Table~\ref{table:dataset} come from different domains and represent Euclidean spaces with dimensions ranging from $8$ to $20$ as well as a wide range of colors (between 4 and $>12{,}000$). The datasets report different levels of balance of color distribution, from complete balance (each color is equally represented in the whole dataset) to high imbalance ($>40\%$ of points of one color).

\begin{table}
\centering
\begin{tabular}{ r | c c c c }
\textbf{Dataset} & \textbf{\# Points} & \textbf{\# Dim.} & \textbf{\# Colors} & \textbf{Max ratio}\\
\hline
\fourarea & $25{,}853$ & 8 &4 &40.2 \\
\query & $>29{,}000$ & 20 & $>12{,}000$ & $<7.0\%$ \\
\reuters &  2500 &10 & 50 & $2.0\%$ \\
\victorian & 4500 & 10 & 45 & $2.2\%$ \\
\hline
\end{tabular}
\caption{Datasets used. Column \# Dim. reports the number of dimensions of the space used and column max ratio represents the maximum ratio of a color in the dataset.} 
	\label{table:dataset}
\end{table}

We now describe more in detail the datasets used.
We obtained two datasets (\reuters, \victorian) from text embeddings of multi-author datasets, one from a co-authorship graph embedding (\fourarea), and one from online auctions (\query). All datasets represent points in the Euclidean space and we always use the $\ell_2$ distance.

\para{(i) \reuters\footnote{Available at \url{archive.ics.uci.edu/ml/datasets/Reuter_50_50}}} It contains 50 English language texts from each of $50$ authors (for a total of $2{,}500$ texts). We transformed each text into a 10-dimensional vector using Gensim's Doc2Vec with standard parameter settings. Here, the colors represent the author of the text. We observe that clustering doc2vec embeddings has been used extensively in language analysis (see, e.g.,~\cite{cha2017language}).

\para{(ii) \victorian\footnote{Available at \url{archive.ics.uci.edu/ml/datasets/Victorian+Era+Authorship+Attribution}}} It consists of texts from 45 English language authors from the Victorian era. Each text consists of $1{,}000$-word sequences obtained from a book of the author (we use the training dataset). The data has been extracted and processed in~\cite{gungor2018fifty}. From each document, we extract a 10-dimensional vector using again Gensim's doc2vec with standard parameter settings and we use the author as color. We use 100 texts from each author.

\para{(iii) \fourarea~\footnotemark[3]} It contains $25{,}853$ points in $8$ dimensions representing each a researcher in one of four areas of CS: data mining, machine learning, databases, and information retrieval. The color is the main area of research of the author. The points are obtained by using the graph embedding method DeepWalk~\cite{deepwalk} on the undirected co-authorship graph of \fourarea, using default settings.

\para{(iv) \query}  It is a representative subset of an anonymized proprietary dataset. Each point in this dataset represents a bag of queries in an online auction environment. The points have $20$ dimensions and are obtained with a proprietary embedding method that encodes semantic similarity. The color of the point is the anonymous id of the main advertiser of the submarket represented by the bag.

\subsection{Experimental setup}

\subsubsection{Baselines}

We use the following two baselines. 

\para{(i) Greedy} Because the $k$-center problem is NP-hard, even without the additional constraint of being $\alpha$-\capped, we use the well-known $k$-center greedy method, which ignores the  representation constraint, as a gold standard. Notice that this algorithm returns a $2$-approximation of the cost of the optimum (without representation constraint) which is always lower than the optimum cost of our problem. To further strengthen the baseline, we post-process the output apply a round of the standard Lloyd iterative algorithm, with $k$-center cost. This step can only improve the results.  We use this method as a gold standard baseline to evaluate the increased cost incurred by our algorithm to enforce the representation constraint and we measure how much our algorithm improves the representation constraint bound of the clusters.

\para{(ii) Random} We also compare against the  baseline of sampling $k$ random points as centers and assigning all points to the nearest center selected. Because this method depends on randomness (while all other algorithms are deterministic), we rerun the algorithm ten times and report the average results. Notice that this algorithm as well does not (necessarily) respect the \capped constraints.

\subsubsection{Measures of quality}

We evaluate the following measures of quality for a clustering.

\para{Cost} We measure the maximum distance of a point to the nearest center in the solution. In particular, we compare the cost of the solution output by our $\alpha$-\capped $k$-clustering algorithm, (for a certain $\alpha$), and the solution of the baselines for the same $k$.

\para{Additive violation of representation constraint}
Recall that our algorithm in Section~\ref{sec:general} can output a solution mildly violating the representation constraint.  We wish to study how big is this violation in practice.  To this end, let $C$ be a cluster in the solution output of an $\alpha$-\capped clustering instance. The maximum allowed number of points of a certain color in the cluster $C$ is $\lfloor |C|\alpha\rfloor$. We let $\Delta = \max_{C,c} \max(|C \cap D_c| - \lfloor |C|\alpha\rfloor, 0)$ be the maximum additive violation of the $\alpha$-\capped constraint, over any cluster $C$ and any color $c$. Our algorithm, provably, has an additive violation $\Delta$ of at most $2$ point. We also evaluate the additive violation of the output of the greedy algorithm and random.

\subsubsection{Implementation details and parameters of the algorithm}

We now describe the main parameters of the algorithm in Section~\ref{sec:general}. The algorithm takes in input $k$, $\alpha$, representing the number of centers allowed and parameter of the $\alpha$-\capped constraint. 
To find a small $\lambda$ for which the polytope $P(\lambda,\alpha)$ gives a feasible solution, instead of binary search, we use following method. We obtain a lower bound on the cost  the clustering by running the greedy $k$-center algorithm and using $\frac{\lambda'}{2}$ as a lower bound, where $\lambda'$ is the cost of the solution found (this is provably a lower bound of the cost for our problem). We also bound the maximum distance of two points by $\lambda''$ (e.g., by using $2$ times the maximum distance of a fixed point to any other point) and iterate over a grid $\Lambda$ that is exponentially increasing by a $(1+\epsilon)$ multiplicative factor between these two extremes, 
$$
\Lambda = \left\{ \frac{\lambda'}{2}, \frac{\lambda'}{2}(1+\epsilon), \frac{\lambda'}{2}(1+\epsilon)^2, \ldots, \lambda'' \right\},
$$
to find the smallest feasible $\lambda$. Notice that a solution is found unless the problem is infeasible (i.e., $\alpha$ is lower than the maximum fraction of points of a color). This allows us to check the LP feasibility with lower $\lambda$'s first, which is better since checking feasibility becomes computationally more expensive as $\lambda$ increases.  

\begin{algorithm}[ht]
\caption{FasterAlgorithm($\cI =(D, F, d, k), \epsilon, m$).}\label{alg:exp}
\begin{algorithmic}[1]
\State $\lambda'' \gets \max_{j \in D} d(i_0,j)$ for arbitrary $i_0 \in D$
\State $((F',\sigma'), \lambda') \gets \textrm{Greedy-$k$-center}(\cI =(D, F, d, k))$
\State $((F^c, \sigma^c), \lambda^c) \gets \textrm{Greedy-$k$-center}(\cI^c =(D, F, d, m * k))$
\For {$\lambda \in \{\frac{\lambda'}{2}, \frac{\lambda'}{2}(1+\epsilon), \frac{\lambda'}{2}(1+\epsilon)^2, \ldots, 2*\lambda''\}$}
\State $(F^s,\sigma^s) \gets \textrm{Fair-$k$-center}(\cI' = (D, F^c, d, k), \lambda)$.
\If{$(F^s,\sigma^s)$ is non-empty}
\Return $(F^s,\sigma^s)$
\EndIf
\EndFor
\end{algorithmic}
\end{algorithm}

Finally, to speed-up the computation, we restrict the variables $y_i$,$x_{ij}$ that we create to be non-zero only for $i \in F' \subseteq F$ where $F'$ is a core-set of the dataset, obtained by running the greedy algorithm to select $m \times k$ facilities. Notice that using $m\ge 1$ results, provably, in a constant factor approximation algorithm. We evaluate the effect of $\epsilon = 0.1, 0.5$, and experiment with $m\ge 2$.

All our computations are run, independently, each on a single machine, from a proprietary Cloud, using Google's Linear Optimization Package (GLOP) as our LP solver, and a maximum flow solver in C++. Both packages are available in Google's OR tools.%
\footnote{\url{https://developers.google.com/optimization/}}

\subsection{Experimental results}

\begin{table}
{\footnotesize
\begin{tabular}{r|cccccc}
\toprule
Dataset & $\alpha$ &  \makecell{Cost vs \\  Greedy} & \makecell{Cost vs \\ Random} &  $\Delta$ &  $\Delta_{\mathrm{G}}$ & $\Delta_{\mathrm{Rand}}$ \\
\midrule
\fourarea & 0.45 &             +62\% &             +50\% &                          1 &                         32 &                            660 \\
      & 0.50 &             +67\% &             +55\% &                          1 &                         19 &                            552 \\
      & 0.60 &             +62\% &             +50\% &                          1 &                          6 &                            338 \\
      & 0.70 &             +64\% &             +52\% &                          0 &                          2 &                            124 \\
      & 0.80 &             +64\% &             +52\% &                          0 &                          0 &                              0 \\
\midrule
\query & 0.07 &              +6\% &              +7\% &                     1 &                        132 &                           66 \\
         & 0.08 &              +6\% &              +7\% &                     1 &                          9 &                           46 \\
         & 0.09 &              +6\% &              +7\% &                     0 &                          7 &                           26 \\
         & 0.10 &              +6\% &              +7\% &                     0 &                          4 &                            6 \\

\midrule
\reuters & 0.02 &             +80\% &             +44\% &                          1 &                         35 &                             38 \\
      & 0.05 &             +75\% &             +40\% &                          1 &                         29 &                             35 \\
      & 0.10 &             +53\% &             +22\% &                          1 &                         24 &                             29 \\
      & 0.20 &              +7\% &             -15\% &                          1 &                         17 &                             18 \\
      & 0.30 &              -3\% &             -23\% &                          1 &                         15 &                             10 \\
      & 0.40 &             +31\% &              +4\% &                          0 &                         12 &                              8 \\
      & 0.50 &              -3\% &             -23\% &                          0 &                          9 &                              6 \\
\midrule
\victorian & 0.05 &            +109\% &             +26\% &                          1 &                         62 &                             57 \\
      & 0.10 &             +45\% &             -13\% &                          1 &                         56 &                             38 \\
      & 0.20 &             +39\% &             -17\% &                          1 &                         43 &                              9 \\
      & 0.30 &             +63\% &              -2\% &                          1 &                         30 &                              0 \\
      & 0.40 &             +45\% &             -13\% &                          1 &                         17 &                              0 \\
      & 0.50 &             +45\% &             -13\% &                          0 &                         10 &                              0 \\
\bottomrule
\end{tabular}
\caption{Comparison of the cost and maximum additive violation of representation constraint for our algorithm, as well as the baselines, over various datasets and $\alpha$ factors, for  $k=25$, $\epsilon =0.1$, $m=2$. We report the ratio of the cost of our algorithm's solution with respect to both the greedy algorithm (Cost vs Greedy) and the random baseline  (Cost vs Random); the maximum additive violation for our algorithm ($\Delta$), the maximum additive violation of the greedy algorithm ($\Delta_{\mathrm{G}}$), and of the random baseline ($\Delta_{\mathrm{Rand}}$).      \label{table:overview}}
}
\end{table}

\paragraph{Comparison with the baselines.}
In Table~\ref{table:overview}, we report, for various $\alpha$ factors, a comparison of the quality of the output of our algorithm with that of the baselines. In this table, we fix the parameters: $k=25$, $\epsilon=0.1$, $m=2$ and show results for all datasets and representative $\alpha$'s that are close to the maximum color ratio of a color in each dataset (there is no feasible solution for $\alpha$'s lower than this ratio).

First, we evaluate the ratio of the cost (i.e., the maximum distance of a point to its center) of the solution obtained by our algorithm to that of the greedy algorithm. Notice that in all datasets our algorithm reports a cost that is relatively close to the unconstrained greedy algorithm and is usually between $+10\%$ worse and up to $2$x worse. Interestingly, despite the fact that the unconstrained problem can have a much better optimum cost, we can sometimes obtain costs that are at most $10$--$50\%$ larger than of the unconstrained solution (which in turn is lower than the actual optimum value for our problem). This result is better than that predicted by the worst-case theoretical analysis (where we show a $3$x factor). This improvement occurs even for $\alpha$ very close to the strongest possible representation constraint for which there is a solution. 

In Table~\ref{table:overview}, we also evaluate the maximum additive violation of the color cap constraint for our algorithms as well as the baselines. As proved formally, the maximum additive violation for our algorithm ($\Delta$) is at most $2$ for general $\alpha$'s (and $1$ for the case of integer $1/\alpha$). We observe interestingly that it is always $1$ in our experiments. Note instead that the baselines, which do not take into account the constraint, can incur very large additive violations of up to hundreds of points. This result confirms the importance of using algorithms specifically designed for this problem.

\paragraph{Effect of the parameters}

We now study more in detail the effect of the main parameters $k, \alpha, \epsilon, m$ on the quality of the clustering.

\begin{figure}
\begin{center}
\begin{tabular}{cc}
\subfigure[Setting $\epsilon = 0.1$]{\includegraphics[width=0.23\textwidth,keepaspectratio]{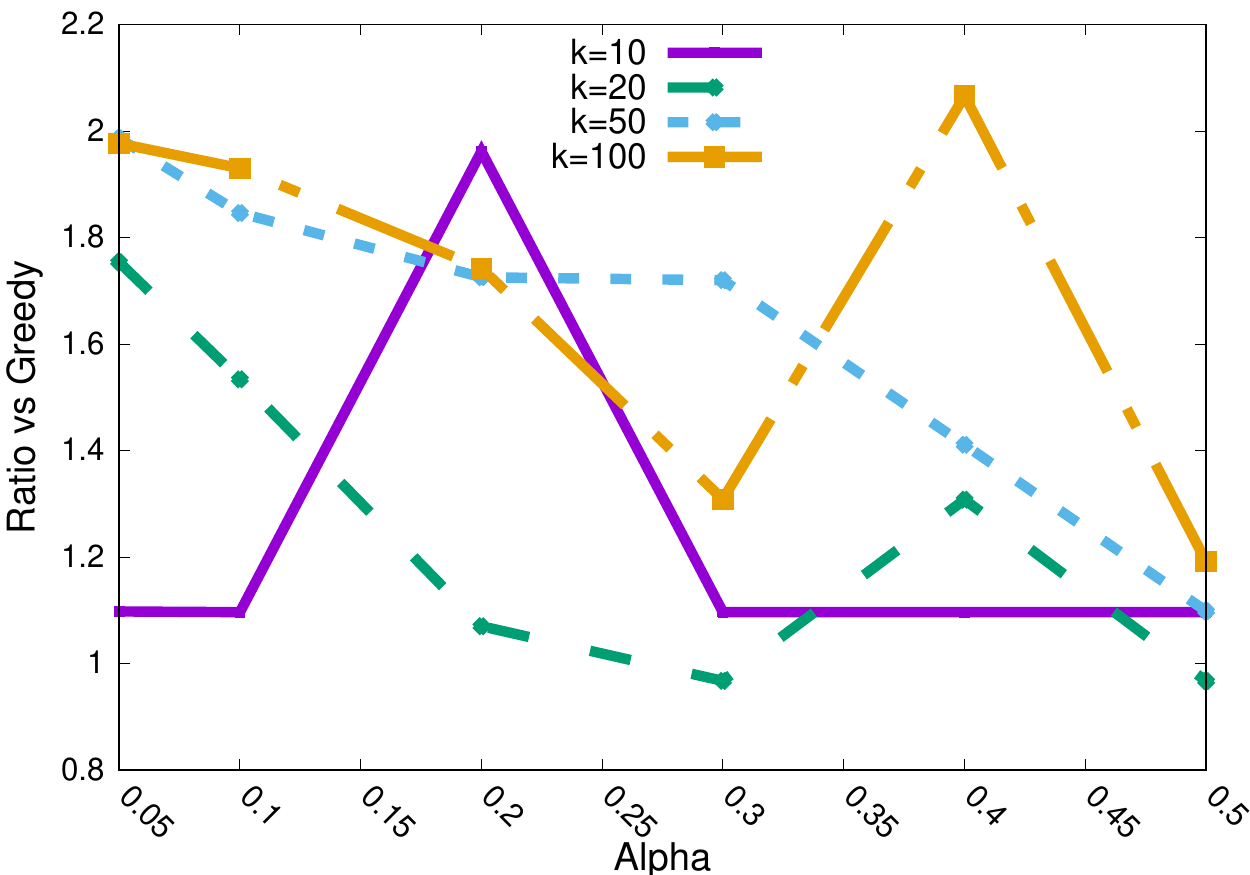}\label{fig:cost-reuters-01}}
\subfigure[Setting $\epsilon = 0.5$]{\includegraphics[width=0.23\textwidth,keepaspectratio]{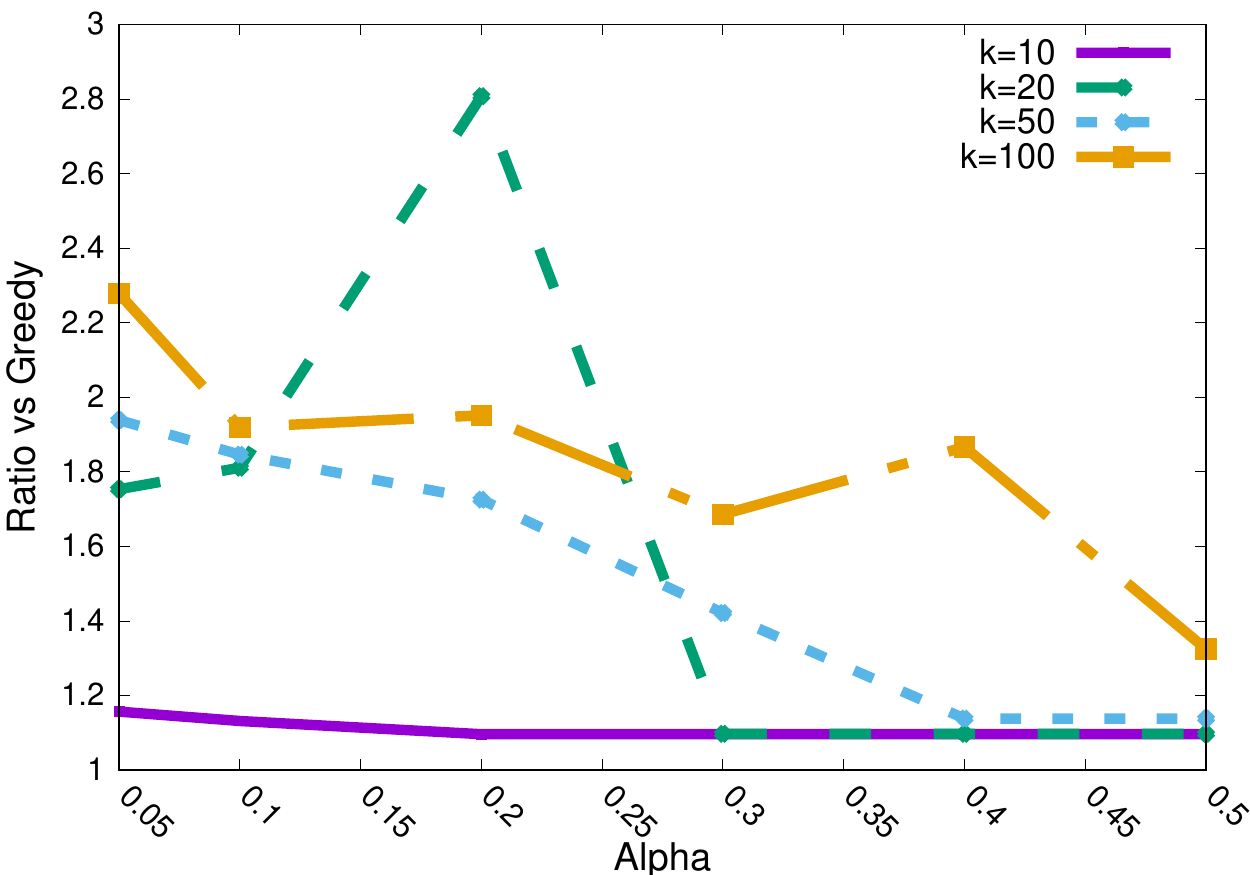}\label{fig:cost-reuters-05}}
\end{tabular}
\caption{Cost of the solution vs Greedy baseline for various $\alpha$, $k$ over the \reuters dataset, using $\epsilon = 0.1,0.5$, $m=2$.}
\label{fig:cost-reuters}
\end{center}
\end{figure}

Figures~\ref{fig:cost-reuters-01} and~\ref{fig:cost-reuters-05} show the ratio of the cost of the solution over the cost of the greedy baseline, for various $\alpha$ ranges, and distinct $k$'s, in the \reuters dataset. Here, we compare the setting $\epsilon = 0.1$ (Figure~\ref{fig:cost-reuters-01}) and $\epsilon = 0.5$ (Figure~\ref{fig:cost-reuters-05}). Notice how the approximation ratio (over greedy) is always $\lessapprox 2$ for the $\epsilon = 0.1$ case and $\lessapprox 3$ for $\epsilon = 0.5$ case. As is expected, notice that larger $\alpha$'s are associated with lower cost ratios (it is easier to find a low cost solution with higher $\alpha$). Finally, despite the pattern being less strong, we observe generally larger ratios for larger $k$'s. 

\begin{figure}
\centering
\begin{tabular}{cc}
\subfigure[]{
\includegraphics[width=0.23\textwidth,keepaspectratio]{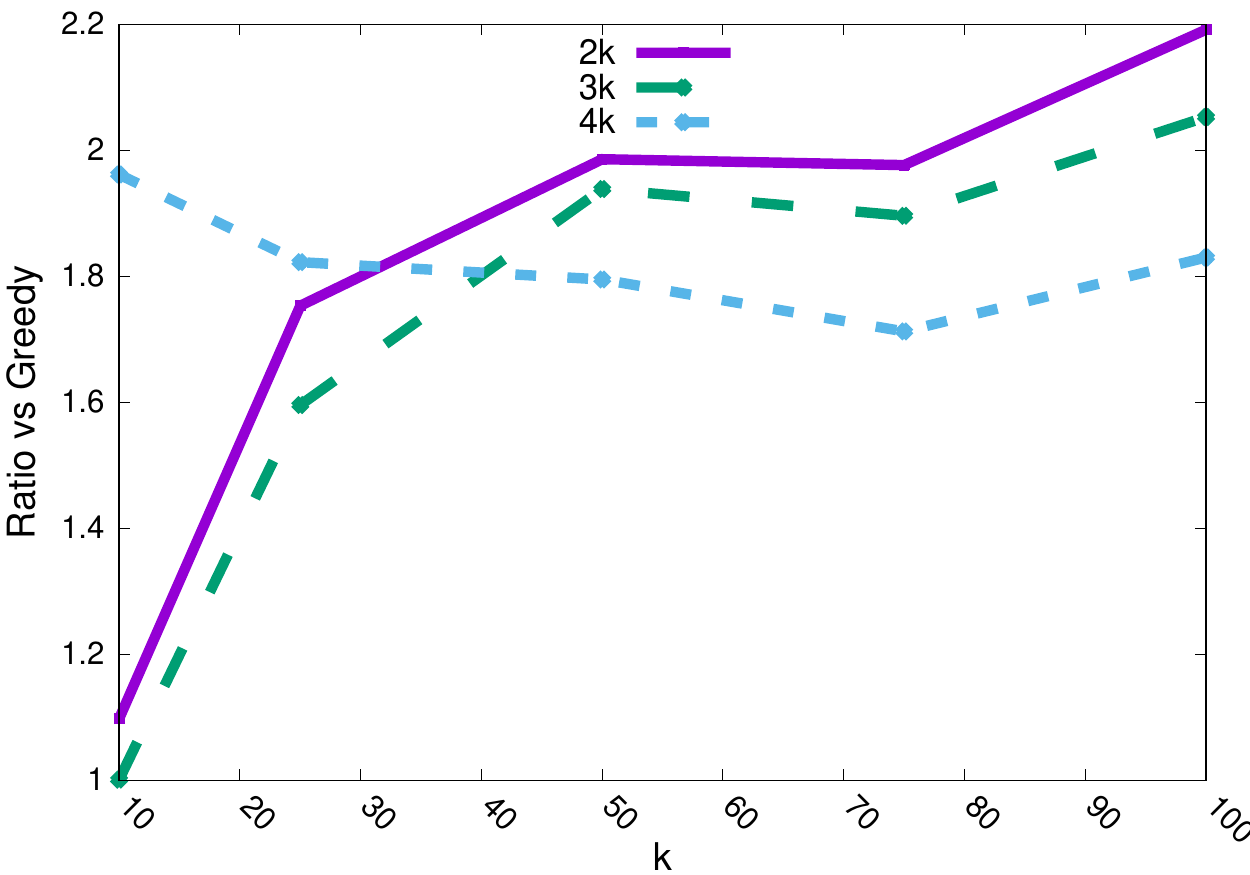}}
\includegraphics[width=0.23\textwidth]{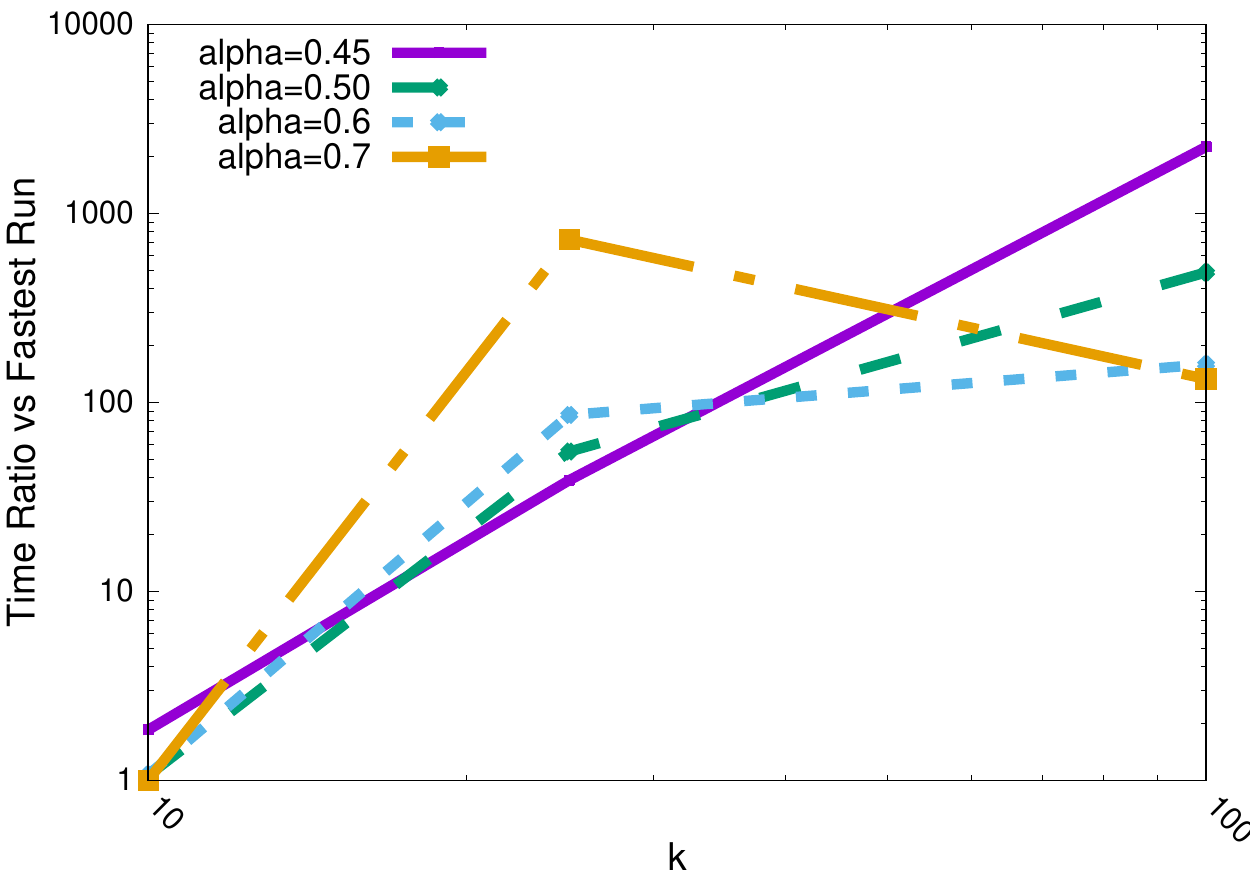}
\end{tabular}
\caption{(a) Cost of the solution vs Greedy baseline for various $k$ and $m$  over the \reuters dataset, using $\alpha=0.05$, $\epsilon = 0.1$. 
(b) Time vs $k$ and $\alpha$ for \fourarea dataset, $\epsilon = 0.5$, $m=2$. We report the ratio of running time of a given instance over the the fastest instance.
\label{fig:merged}
}
\end{figure}

In Figure~\ref{fig:merged}(a) we evaluate the effect of the $m$ factor used in the core-set to reduce the number of $y_i$'s variables to $m \times k$. Notice how generally larger $m$'s are associated with lower cost (ratio), but the algorithm obtains good results even with $m = 2$, allowing to use small LP instances in our algorithm. 

\paragraph{Time}

In Figure~\ref{fig:merged}(b) we show how the running time is affected by $k$ and $\alpha$. As expected, larger $k$'s correspond to increased running times.  Similarly, larger $\alpha$'s mostly correspond to lower running time because it is easier to find a solution with larger $\alpha$ and hence fewer $\lambda \in \Lambda$ need to be evaluated to find a non-empty $\cP(\lambda, \alpha)$.
\newcommand{\tred}{\texttt{red}\xspace}
\newcommand{\tblue}{\texttt{blue}\xspace}

\section{Hardness}
\label{sec:hardness}
In this section, we complement our algorithmic results by proving a factor-$2$ approximation hardness for  minimizing the $k$-center cost of a $\alpha$-capped clustering, of arbitrary number of cluster,  for $\alpha \in (0,0.5]$. This shows the hardness of $\alpha$-capped clustering, with $k$-center objective, even allowing arbitrary many clusters.

As in~\cite{chierichetti2017fair}, we use a reduction from the \tsD\xspace problem defined as follows. Given an undirected $n$-vertex graph $G=(V,E)$, and a positive integer $t$, can $V$ be partitioned into pairwise disjoint subsets $V_1, \ldots, V_{n/t}$ so that $|V_i| = t$ and $G[V_i]$ contains a star of size $t$, i.e., a center and $t-1$ leaves? Two well-known special cases of \tsD are the case $t=2$ (finding a perfect matching) and the case $t=3$ also known as $P3$-decomposition (finding a partition into connected triplets). Since a perfect matching  can be found in polynomial time, \tsD is tractable for $t=2$.  Kirkpatrick and Hell~\cite{kirkpatrick1983complexity} showed that \tsD is NP-hard for $t \geq 3$.  \tsD  remains NP-hard~\cite{dyer1985complexity} even if the graph is planar and bipartite, for any $t\ge 3$. In our proofs we will use that the problem is NP-hard.

Our reduction starts from input $G$ of a \tsD\xspace instance, and defines a set $D$ of points in a metric space with distance function $d(\cdot, \cdot)$ and a color assignment $c(j)$ for each point $j\in D$. More precisely, we construct a graph $G' = (D, E')$ and define the metric space to be the shortest path metric where edges have unit length. Before proceeding to the main hardness result, we explain how graph $G'$ is constructed in polynomial time from the bipartite  graph $G = (V_1\cup V_2,E)$ input of \tsD. In the following we use the word point and vertex interchangeably.

\paragraph{Construction of the graph $G'$} The construction of $G'$ depends on the solution to the following system of linear equations:
%\begin{eqnarray}
%2 t_{r} + 1 t_{b} &=& |V_1| %\label{eq:system1}\\
%1 t_{r} + 2 t_{b} &=& |V_2| %\label{eq:system2}   \end{eqnarray}
\begin{equation}\label{eq:system1}
2 t_{r} + 1 t_{b} = |V_1| \qquad \textrm{and}\qquad 1 t_{r} + 2 t_{b} = |V_2|  
\end{equation}
%\label{eq:system1}
%\label{eq:system2}   \end{eqnarray}
Since this is a system of two equations in two variables, and the determinant of the system is non-zero, there exists a unique solution $(t_r, t_b)$. If the unique solution has at least a variable that is not a non-negative integer, we construct $G'$ as a trivial instance with no fair coverage (say one \tred node). For the rest of the construction we assume we are in the case that $t_r$, $t_b$ are both non-negative integers.

First we define the construction for the $\alpha=\frac{1}{2}$, case then we show how to extend this to the $\alpha=\frac{1}{2+t}$ case for any integer $t>0$. In the $\alpha=\frac{1}{2}$ case, the construction proceeds as follows. The graph $G'=(V',E')$ has four layers of nodes $L_1,L_2,L_3,L_4$, where each layer $L_i$ consists of two disjoint sets $R_i,B_i$ of respectively of color \texttt{red} and \texttt{blue}. The layer $L_1$ has a 1-to-1 correspondence with nodes in $V$. More precisely, $L_1$ consists of $R_1 \equiv V_1$ and $B_1 \equiv V_2$, corresponding to the two sides of the graphs $G$ and two nodes in $L_1$ are connected in $E'$ iff their equivalent nodes are connected in $E$. Then, $L_2$ consists of $R_2,B_2$ such that $|R_2|=|R_1|, |B_2|=|B_1|$. In $E'$, there is a matching between each node in $R_2$ (resp. $B_2$), and a node in $R_1$ (resp. $B_1$). Now let $u_b = |B_2| - t_r$ and $u_r = |R_2| - t_b$. Notice that from the Equations~(\ref{eq:system1}) $u_b,u_r$ are non-negative integers. Layer $L_3$ has components $B_3, R_3$ of size $|B_3| = u_b$, $R_3 = u_r$ and $E'$ contains a complete bipartite graphs between sides $R_2$,$R_3$ and another complete bipartite graph between sides $B_2,B_3$. Finally layer $L_4$ consists of $R_4,B_4$ such that $|R_4|=2|B_3|$ and $|B_4|=2|R_3|$ and each node in $R_3$ is connected with exactly two nodes in $B_4$ (resp. each node in $B_3$ is connected with exactly two nodes in $R_4$). This completes the construction for the $\alpha=\frac{1}{2}$ case, for the general $\alpha=\frac{1}{2+t}$ we add to each layer $L_2$ and $L_4$, $t$ disjoint sets $C^t_i$ ($i=2,4$) such that all nodes in $C^t_i$ have color $c_t$ (distinct from \tred and \tblue). For each $t$, $|C^t_2|= 2 (t_r+t_b)$ and $C^t_2$ is further subdivided in two disjoint parts $C^t_{2,r}$, $C^t_{2,b}$ such that $|C^t_{2,b}|=2t_b$, $|C^t_{2,r}|=2t_r$, and  $B_1,  C^t_{2,b}$ for a complete bipartite graph (reps.   $R_1,  C^t_{2,r}$ form a complete bipartite graph).
Finally for each $t$, $|C^t_4| = 2 |L_3|$ and each node in $L_3$ is connected with exactly 2 nodes in $C^t_4$.

The following states our main hardness result for $\alpha \in (0, 0.5]$.

\begin{theorem}
It is NP-hard to approximate the $\alpha$-capped clustering with $k$-center objective with $\alpha \in (0,0.5]$ within a factor better than $2$.
\end{theorem}
The theorem follows from the following two lemmas,  whose proofs are deferred to the extended version of the paper. %in Appendix~\ref{sec:app}.

\begin{lemma}
\label{lem:6.2}
Fix $t \ge 0$ integer. Suppose the bipartite graph $G$ admits a \tsD, then $G'$ has a $\frac{1}{2+t}$-capped clustering of $k$-center cost $1$.
\end{lemma} 

\begin{lemma}
\label{lem:6.3}
Fix $t \ge 0$ integer. If there exists a solution of $k$-center cost at most $2$ to $\frac{1}{2+t}$-capped  clustering of  $G'$, then the bipartite graph $G$ admits a \tsD.
\end{lemma}

\section{Conclusions}
\label{sec:conclusion}

Clustering with color constraints is an algorithmic take on ensuring balance and fairness in applications.  In this paper we addressed \capped clustering, which is the problem of finding the best clustering where no cluster has an over-represented color.   We obtained provably good algorithms for this problem; our experiments show that the algorithms are effective on different real-world datasets.  While our general algorithm is based on solving an LP, it can be challenging for large number of points.  It is an interesting question to develop a combinatorial algorithm for the general case that can scale to large datasets.  It is also interesting to improve the bounds guaranteed by our algorithms and extend them to other clustering objectives such as $k$-means and $k$-median.  

\balance
\bibliographystyle{ACM-Reference-Format}
\bibliography{paper}

\end{document}